\documentclass[hidelinks]{article}
\usepackage{amsmath,amssymb,amsthm,bbm,theoremref}
\usepackage{booktabs, csquotes, graphicx, subfig, mathtools}
\usepackage{bookmark, caption, float}

\def\reals{\mathbbm{R}}

\def\A{{\cal A}}
\def\C{{\cal C}}
\def\F{{\cal F}}

\def\M{{\cal M}}
\def\V{{\cal V}}

\def\reals{\mathbbm{R}}
\def\ereals{\overline{\reals}}
\def\minimize{\mathop{\rm minimize}\limits}
\def\ovr{\mathop{\rm over}\ }
\def\st{\mathop{\rm subject\ to}}

\def\midd{\,|\,}

\newtheorem{theorem}{Theorem}

\newtheorem{proposition}[theorem]{Proposition}

\newtheorem{remark}[theorem]{Remark}

\title{Optimal Pricing and Hedging of SOFR Derivatives\thanks{The authors would like to thank Damiano Brigo, Thorsten Schmidt and Wahid Khosrawi for helpful comments and suggestions during the project.}}
\author{Teemu Pennanen\thanks{Department of Mathematics, King's College London, Strand, London, WC2R 2LS, United Kingdom, teemu.pennanen@kcl.ac.uk} \and Waleed Taoum\thanks{Department of Mathematics, King's College London, Strand, London, WC2R 2LS, United Kingdom, waleed.taoum@kcl.ac.uk}}
\date{\today}

\begin{document}

\pdfbookmark{Table of Contents}{name} 

\maketitle

\begin{abstract}
Thousands of SOFR derivatives are available in exchanges and OTC, but the market remains illiquid and incomplete. Such a market is beyond the scope of classic risk-neutral approaches that imply linear pricing rules and, at best, approximate hedging strategies whose hedging error may be difficult to quantify. This paper develops an indifference pricing model which is consistent with observed derivative quotes, the agent's financial position and views about the uncertain future as well as risk preferences as described by a convex risk measure. In addition to prices and hedging strategies, the model gives an explicit description of the hedging error and the associated risk. The approach is illustrated numerically using hundreds of CME-listed derivatives to price and hedge unreplicable OTC SOFR derivatives. The indifference prices are computed in less than a minute on a regular PC. We find that the optimal hedging portfolios tend to be sparse but still provide good approximations of the derivative payouts.
\end{abstract}

\section{Introduction}\label{sec:introips}

In a complete financial market, the payouts of all contingent claims can be replicated by the proceeds of some trading strategy. In such a market, the unique arbitrage-free price of a claim is given by the replication cost, as beautifully illustrated in~\cite{bs73} in a continuous-time market model driven by a Brownian motion. Soon after the publication of~\cite{bs73}, it was found that replication-based prices of contingent claims could be expressed as expectations of their cashflows under so-called ``risk-neutral measures"; see~\cite{hk79}. This was an elegant application of convex analysis, and it soon became the standard way to view derivative pricing. The result, often called the ``fundamental theorem of asset pricing" has, however, had the unfortunate effect of shifting attention from the underlying economics (the hedging problem) to more mathematical issues in stochastics (the study of martingales). This has turned out to be particularly cumbersome in incomplete markets where replication is usually impossible. Indeed, risk-neutral pricing says little about how to best hedge unreplicable instruments. Moreover, it results in a linear pricing rule which is at odds with the nonlinear illiquidity effects observed in practice. 

Even when exact replication is not possible, it still makes sense to look at trading and hedging costs when pricing financial products. In incomplete markets, one just has to reconsider what kind of hedges and prices would be acceptable for a seller or a buyer of a product whose cashflows cannot be perfectly replicated. This paper studies the {\em indifference pricing} principle in the context of SOFR derivatives markets where incompleteness arises both from illiquidity effects as well as random jumps in the rates caused by monetary policy decisions.

In indifference pricing, one treats investment decisions, risk management, derivative pricing and hedging a single unifying framework. The prices are consistent with observed market prices of available financial instruments as well as the agent's existing financial position, risk preferences and views concerning the uncertain future. The approach constructs optimal hedging strategies along with a quantification of the hedging error that may be significant in incomplete markets. If an instrument happens to be replicable, its indifference price coincides with its replication cost in accordance with classical replication arguments of~\cite{bs73}. Further connections with risk-neutral valuations can be seen in the {\em dual} of the optimal investment problem. Indeed, the variables in the dual problem turn out to be martingale measures or, more generally, ``consistent price systems/stochastic discount factors" calibrated to the observed market prices of the available hedging instruments; see~\cite[Example~3.75 and Remark~3.76]{pp24}. Indifference pricing is quite a general approach, and it can be applied in arbitrary market sectors and market models. References on indifference pricing include~\cite{car9,hodges1989optimal,brigo17,pen14,armstrong2018pricing,mz4,pen26}. The indifference pricing principle can be found already in~\cite[Section~4.1.2]{buh70} under the name of ``principle of zero utility". To our knowledge, the present paper is the first one to study indifference pricing in SOFR derivatives markets. 

Various risk-neutral pricing models for SOFR derivatives have been proposed in the recent years. Most of the models can be seen as modifications of classical interest rate models found e.g.~in~\cite{brigo2006interest,fil9,mr5}. The most common approach seems to be to fix one martingale measure and to use that for pricing of all derivatives; see e.g.~\cite{mercurio2018simple,lyashenko2019libor,andersen2020spike, macrina20,skov2021dynamic,xu21,FON2023,ls23,russo23,rf23,Turfus2023,fusai24,IK2024,lS24,schlogl2024term,yueh24,calvia2026short,shai2026}. Such an approach is often computationally convenient, but it doesn't say much about how to best hedge a general interest rate derivative. Certain simple interest rate products such as floating rate loans, swaps, and treasury bonds can be replicated by portfolios of SOFR futures; see e.g.~\cite{fabozzi24,huggins22,CMErep}. Lyashenko and Mercurio~\cite{lyashenko2020libor} as well as Bickersteth et al.~\cite{BDR2025} impose assumptions that make their market models complete, thus avoiding the problems that come with incompleteness. On the other hand, their models do not allow for jumps in the rates which are often observed at FOMC meeting dates. Fontana et al.~\cite{fontana2024term} allow for stochastic discontinuities and construct approximate hedging strategies for nonreplicable derivatives by quadratic hedging approach that is surveyed in~\cite{MS2001}. Unlike indifference pricing with a convex risk measure, however, quadratic hedging does not distinguish between profits and losses.

Indifference pricing can be built on any model of optimal trading. This paper develops a semi-static trading model where one optimizes a static portfolio in exchange-traded derivatives while all payouts are carried forward in the overnight money market. More specifically, we consider SOFR derivatives listed by CME. CME lists 13,500 SOFR derivatives, all of which come with bid-ask spreads and finite quantities. In fact, almost half of them have zero liquidity at the time of writing. The indifference prices and optimal hedges are computed in less than a minute on a regular PC. The optimal hedging portfolios tend to be sparse but still provide good approximations of the derivative payouts. We also study the sensitivity of the prices with respect to the agent's risk preferences and the liquidity of the hedging instruments. Increasing the bid-ask spreads results in even sparser hedging portfolios, while the sensitivity with respect to the risk aversion seems to be strongly dependent on the derivative contract being priced. 

The remainder of the paper is organised as follows. Section~\ref{sec:sdm} provides a brief description of the CME-listed hedging instruments as well as the OTC derivatives studied in Section~\ref{sec:ipsd}. Section~\ref{sec:rf} gives a quick overview of the term structure model that will be used in the computations of this paper. Section~\ref{sec:PFO} describes the optimal investment problem that will be used to define the indifference prices later on. Section~\ref{sec:optimum} studies the optimum value of the optimization problem as a function of the agent's financial position. Section~\ref{sec:iph} defines the indifference prices and studies their basic properties. Section~\ref{sec:ipsd} presents a numerical study of applying indifference pricing to OTC SOFR derivatives.

\section{The SOFR Derivatives Market}\label{sec:sdm}

This section briefly reviews the derivative contracts studied in the rest of this paper. For further details, we refer the reader to~\cite[Section~2]{Pen25sofr} and its references.

\subsection{Exchange-Traded Derivatives}
\label{sec:exchange}

The CME Group lists SOFR futures, SOFR options, and SOFR swap futures for trading in double auction markets. SOFR futures with one- and three-month reference periods are the underlying of one- and three-month SOFR options, respectively. Both SOFR futures and options have monthly and quarterly maturities. In addition, the options have weekly maturities. 
At the time of writing, CME lists 13,500 contracts in total, but only about 7,300 of them have quotes. We will focus on the most liquid of the listed derivatives: the three-month SOFR futures and futures options, which expire in the March quarterly cycle. It would be straightforward to include more derivatives in the optimization and pricing models below. On 28 August 2024 at 15:30:00, CME listed three-month futures options had quotes for 8  maturities, each with up to 139 strikes. Additionally, there are quotes for 24 three-month futures maturities. 

The floating leg of a SOFR three-month futures contract pays a multiple of the {\em geometric SOFR average} defined by
\begin{equation}\label{eq:SOFR_Average}
R(t_0,t_1) := \frac{1}{(t_1-t_0)\delta} \left[ \prod_{t=t_0}^{t_1-1}\left(1 + r_t\delta\right) - 1\right],
\end{equation}
where $t_0$ and $t_1$ denote the start and end dates of the three-month reference period, $\delta:=1/360$, and $r_t$ is the {\em SOFR rate} applicable from calendar day $t$ to $t+1$; see~\cite[Section 2.1]{Pen25sofr}. In a futures contract, the counterparties agree to exchange a multiple of the SOFR average for a fixed payment at time $t_1$. 

The CME SOFR futures are quoted in terms of {\em futures prices}. 
The payout of a long position in a futures contract with futures price $f$ is given by
\begin{equation}\label{eq:3mfutpo}
100[1 -R(t_0,t_1)]-f
\end{equation}
multiplied by a notional of \$2,500; see~\cite[Section~46003.A]{CMESOFRFut3m}. 
Defining the {\em futures rate} $F=1-f/100$, the payout of a long futures position becomes
\begin{equation}\label{eq:futpor}
    F-R(t_0,t_1)
\end{equation}
multiplied by \$250,000. We will work in terms of futures rates for simplicity.

\begin{remark}[Marking-to-market]
CME SOFR futures are marked-to-market at the end of each trading day; see~\cite[Chapter~8]{CME2022Rulebook}. Denoting the daily futures rate by $F_t$, the net cashflows from marking-to-market a long futures position equal, in agreement with~\eqref{eq:futpor},
\[
\sum_{s=t+1}^{t_1}(F_{s-1} - F_s) = F_t-F_{t_1},
\]
where $F_{t_1} = R(t_0, t_1)$. If the margin account is paid a significant interest, the timing of the payments matters. In a market-neutral portfolio, however, one can expect the net daily payments to be small. The optimization model in Section~\ref{sec:PFO}, below, is aimed at constructing such portfolios, so we will ignore the margining in what follows.
\end{remark}

CME lists three-month futures contracts for 46 different maturities. At any given time, the best available prices for taking a position in a futures contract are given by a {\em limit order book}. In this paper, we will only consider the best price levels. The best {\em futures ask price} available at time $t$ for reference period $[t_0,t_1]$ will be denoted by $f^a_t(t_0,t_1)$. The best {\em futures bid price} will be denoted by $f^b_t(t_0,t_1)$. While entering a futures position is costless, the payout at time $t_1$ of a position of size $x$ is given by
\begin{equation}\label{eq:fp}
\begin{cases}
[F_t^a(t_0,t_1)-R(t_0,t_1)]x & \text{if $x\ge 0$},\\
[F_t^b(t_0,t_1)-R(t_0,t_1)]x & \text{if $x\le 0$},
\end{cases}
\end{equation}
where
\begin{align*}
F_t^a(t_0,t_1) &:= 1-f^a(t_0,t_1)/100,\\
F_t^b(t_0,t_1) &:= 1-f^b(t_0,t_1)/100.
\end{align*}
Since the ask price $f^a_t(t_0,t_1)$ is always strictly higher than the bid price $f^b_t(t_0,t_1)$, we have $F^a_t(t_0,t_1)<F^b_t(t_0,t_1)$ so the futures payout is {\em concave} as a function of the position~$x$.

A CME futures {\em call (put) option} gives its owner the option to enter a long (short) position in a futures contract with the fixed leg equal to the option {\em strike}. We will study options on three-month futures. The reference period of the futures begins on the Wednesday following option expiry; see~\cite{CMESOFROPT3m}. The CME futures options are American style, but the time of exercise does not affect the resulting futures position. Challenges with valuation and optimal exercise of American options on SOFR futures have been studied in~\cite{IK2024}.

\begin{remark}[Early exercise in long position] \label{rem:optionexerc}
Even though the CME futures options are American-style, they can be treated as European options. Indeed, since the time of exercise does not affect the resulting futures position, the option holder is better off waiting until maturity before deciding whether to exercise.
\end{remark}

Assuming that the option holder closes the futures position immediately after exercise at the available market rate, allows us to treat the options as cash-settled instruments. According to CME, closing the the futures position is indeed common in practice; see~\cite{CMEoptex}. If the holder of a long position in a futures call option with strike~$X$ decides to exercise the option at maturity~$t_M$ and immediately closes the position by taking the opposite futures position at the prevailing market futures rate $F_{t_M}(t_0,t_1)$, their net position would pay a multiple of~$X-F_{t_M}(t_0,t_1)$ at the futures maturity~$t_1$. Assuming the option holder is rational and only exercises if the payout is positive, the option can be treated as a cash-settled instrument with payout 
\begin{equation}\label{eq:callpo}
[X - F_{t_M}(t_0,t_1)]^+
\end{equation}
at time $t_1$. A similar argument implies that the payout of a futures put option is given by 
\begin{equation}\label{eq:putpo}
[F_{t_M}(t_0,t_1) - X]^+.
\end{equation}

\begin{remark}[Early exercise in short position]
Even though CME futures options are American style, it is safe for an option seller to assume that the option holder won't exercise early. Indeed, if the holder of a call option decides (against the logic explained in Remark~\ref{rem:optionexerc}) to exercise early, the seller can close their short futures position at the option maturity $t_M$ so that their net liability cashflow would become $X - F_{t_M}(t_0,t_1)$. This is dominated by ~\eqref{eq:callpo} which assumes European-style exercise. An analogous argument applies to futures put options.
\end{remark}

If the price at time $t_M$ of a zero-coupon bond with maturity $t_1$ is $P_{t_M}(t_1)$, we can treat the call option also as an instrument with payout 
\[
P_{t_M}(t_1)[X-F_{t_M}(t_0,t_1)]^+
\] 
at the option maturity $t_M$. Analogously, the put option payout at time $t_M$ would be given by 
\[
P_{t_M}(t_1)[F_{t_M}(t_0,t_1)-X]^+.
\]

\subsection{OTC Derivatives}\label{sec:otc}

The most liquid OTC derivatives are the SOFR overnight index swaps (OIS). An OIS is a swap contract that has annual payments at $T_k$ with $k=1, \ldots, K$. For maturities less than a year, the exchange of cash flows occurs at maturity. The fixed leg payments are equal to the swap rate $X$ multiplied by the day count fraction i.e.\ $X(T_k-T_{k-1})\delta$, while the floating leg is the geometric SOFR average multiplied by the day count fraction i.e.\ $R(T_{k-1},T_k)(T_k-T_{k-1})\delta$; see~\cite{CMEotcswaps}. The fixed leg can be replicated by a collection of zero-coupon bonds while the floating leg can be replicated by rolling money in the overnight market. Assuming that the ZCB for relevant maturities can be traded at prices $P_t(T_k)$, the value at time $t\le T_0$ of a short position in the OIS would be
\begin{equation}\label{eq:ois}
X\sum_{k=1}^K P_t(T_k) \delta(T_k-T_{k-1}) - [P_t(T_0)-P_t(T_K)];
\end{equation}
see e.g.~\cite[Section~2.2]{Pen25sofr}.

Swaptions are call and put options on OIS. If exercised, a call swaption, with strike $X$ and maturity $T_0$, converts to a short position in an OIS with a swap rate $X$; see~\cite{CMEIRprod, CMEswaption}. Assuming perfectly liquid zero-coupon bonds, the market price at the time $T_0$ of such an OIS is given by~\eqref{eq:ois} with $t=T_0$. If the price is negative, a rational agent does not exercise, so the swaption value at maturity $T_0$ becomes
\begin{equation}\label{eq:swaptionpo4}
\left(X \sum_{k=1}^{K} P_{T_0}(T_k) \delta (T_{k} - T_{k-1}) - P_{T_0}(T_0) + P_{T_0}(T_K)\right)^+,
\end{equation}
where $P_{T_0}(T_0)=1$. 

A {\em caplet} with a reference period $[t_0,t_1]$ is a call option on the SOFR average $R(t_0,t_1)$ defined in~\eqref{eq:SOFR_Average}. Thus, the payout of a caplet with reference period $[t_0,t_1]$ and strike $X$ is given by
\begin{equation}\label{eq:capletpo}
[R(t_0,t_1) - X]^+
\end{equation}
at time $t_1$. Similarly, a {\em floorlet} payout is given by 
\begin{equation}\label{eq:floorletpo}
[X-R(t_0,t_1)]^+.
\end{equation}
{\em Caps} and {\em floors} are sequences of caplets and floorlets. 
Other OTC products whose payouts are functions of the future path of the term structure include constant maturity swaps, CMS spread options, barrier and digital options, but they won't be studied in this paper.

\section{A Stochastic Model of the Risk Factors}\label{sec:rf}

To describe the future payouts of the derivative contracts discussed in the previous section, requires a model for the prevailing SOFR rates, futures rates and the prices of zero-coupon bonds. We will model them with a stochastic {\em term structure model} that, at any given time $t$, specifies the zero-coupon prices $P_t(T)$ for maturities $T>t$. The model assumes that the zero-coupon bonds exist for each daily maturity and that they are liquid in the sense that their buying and selling prices are equal. Under these assumptions, the unique arbitrage-free price at time $t$ of a long position in a futures contract with futures rate $F$ would be
\begin{equation}\label{eq:futures}
\frac{P_t(t_0)-P_t(t_1)}{(t_1-t_0)\delta} - FP_t(t_1);
\end{equation}
see e.g.~\cite[Section~2]{Pen25sofr}. Thus, the unique futures ``market rate" that makes the value in~\eqref{eq:futures} equal to zero is given by
\begin{equation}\label{eq:futzcb}
F_t(t_0,t_1) = \frac{P_t(t_0)-P_t(t_1)}{P_t(t_1)(t_1-t_0)\delta}.
\end{equation}

The evolution of the term structure will be described by the stochastic model developed in~\cite{Pen25sofr}. The model describes the daily evolution of the SOFR and the overnight forward rates. The model allows for random jumps on the monetary policy committee meeting dates consistent with the long-term evolution of FED policy rates together with price inflation and GDP. Much like in the Hull-White model~\cite{hull1990pricing}, the median of the overnight rate is allowed to be time-dependent. In~\cite{Pen25sofr}, the median path is taken to be the forward curve calibrated to the observed futures quotes. Importantly, both the macroeconomic factors and the term structure evolution are described under the {\em subjective} probability measure, also known as the ``physical/real-world/historical" or the ``$P$-measure". Instead of the model of~\cite{Pen25sofr}, one could use any subjective term structure model in the optimization and pricing models below.

As a quick illustration of the term structure model, we simulate $2^{16}=65,536$ scenarios of the SOFR forward curve over a two-year period starting on 28 August 2024. The top plot in Figure~\ref{fig:allstrikes} describes the development of the median and the 90\% and 99\% confidence bands of the three-month SOFR average $R(t_0,t_1)$. Since the term structure model is calibrated to the forward curve estimated from the observed three-month futures quotes (see~\cite[Section 4.5]{Pen25sofr}), the median tracks the futures quotes fairly closely. The small vertical lines crossing the median are the futures bid-ask rate intervals. The futures quotes were downloaded from Bloomberg on 28 August 2024 at 15:30:00.

The lower plot in Figure~\ref{fig:allstrikes} describes the development of the median and the $90\%$ and 99\% confidence bands of the model-based quarterly three-month futures rates at option maturity computed from the simulated forward curve using equation~\eqref{eq:futzcb}. The plot also illustrates the strikes of all the listed three-month futures options on 28 August 2024 at 15:30:00. In total, there were 1291 listed options within the two-year maturity. Many of the listed strikes go deeply in the negative territory because of the current CME listing rules. The deep in/out-of-the-money options have low liquidity. In particular, the deep in-the-money options often come with unit bid and ask sizes, which, according to CME, corresponds to market maker offers known as ``cabinets''. On the other hand, many out-of-the-money quotes have zero bid price. In the optimization and pricing models below, we will only use options having strikes within the $90\%$ confidence band of the simulated three-month futures rates. This gives 619 options in total.

\begin{figure}[!ht] 
	\centering 
    \subfloat
    {  
		\includegraphics[trim = 0mm 0mm 0mm 0mm, clip, width=0.79\textwidth, height=0.45\textwidth] 
        {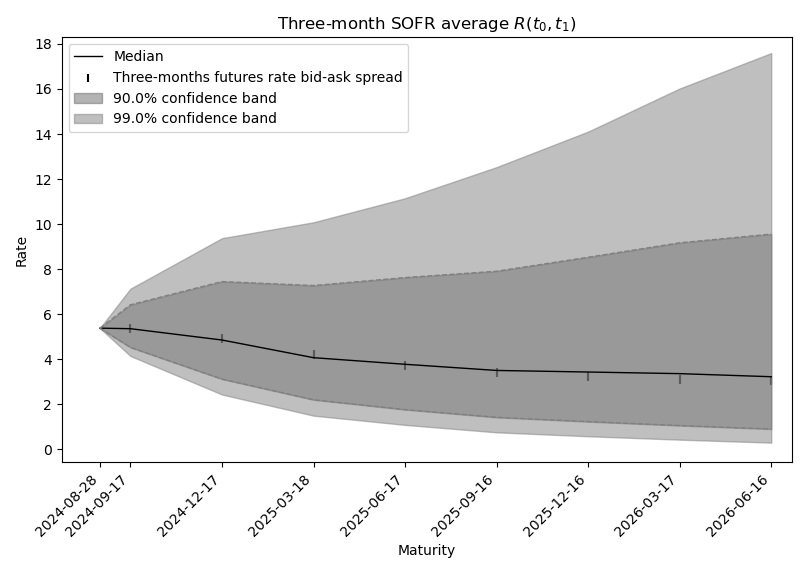}
		\label{fig:allstrikesa}
	}\\
    \subfloat
	{   
		\includegraphics[trim = 0mm 0mm 0mm 0mm, clip, width=0.79\textwidth, height=0.45\textwidth] 
		{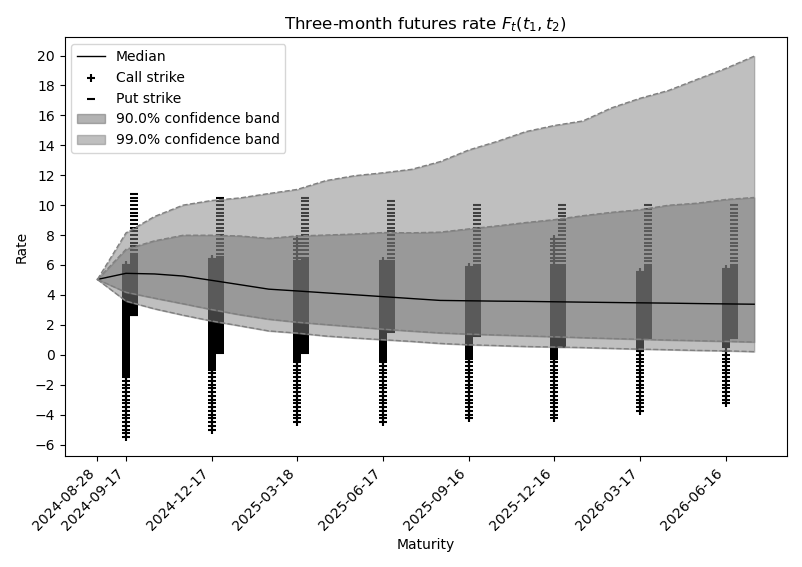}
		   \label{fig:allstrikesb}
	}
	\centering
	\caption
	{The quarterly development of the medians and 90\% and 99\% confidence bands of the simulated values of the underlyings of the SOFR derivatives in Section~\ref{sec:exchange}. The upper plot compares simulated three-month SOFR averages against the three-month futures quotes (the vertical bars). The lower plot compares the simulated three-month futures rates against the strikes of the listed options on three-month futures. The quotes and the strikes were obtained on 28 August 2024 at 15:30:00. Rates are expressed in percentages.
	}
	\label{fig:allstrikes}
\end{figure}

\section{The Optimization Problem}\label{sec:PFO}

Given a set of quotes for the three-month SOFR futures and options together with a stochastic term structure model, we aim to find a buy-and-hold portfolio that yields an optimal wealth distribution at the end of a given investment horizon. In addition to the derivative payouts, the agent may be faced random liability payments associated with their existing financial position. Incorporating such liability payments in the portfolio optimization problem allows us to define indifference prices in terms of the optimum value of the problem; see Section~\ref{sec:iph} below. Throughout, all uncertainties are modelled by the stochastic model of Section~\ref{sec:rf}. We will occasionally denote the corresponding probability space abstractly by $(\Omega,\F,P)$.

We denote by $J$ the collection of quoted derivatives. In the numerical computations below, we will only consider the best price levels, so that the cost of buying $x^j$ units of derivative $j\in J$ is
\begin{equation*}\label{eq:TransactionCosts}
S_0^j (x^j) = 
\begin{cases} 
s^j_a x^j & \text{if } x^j\ge 0,\\
s^j_b x^j & \mbox{if } x^j\le 0,
\end{cases}
\end{equation*}
where $s^j_a\ge s^j_b$ are the ask and bid price, respectively, of derivative $j$. As usual, negative $x^j$ is interpreted as shorting $-x^j$ units of the contract. The total cost of a portfolio $x\in\reals^J$ is given by
\[
S_0(x) = \sum_{j\in J} S_0^j(x^j).
\]
Possible proportional transaction costs can be incorporated simply by reducing the bid prices and increasing the ask prices.

We will denote by $q^j_a$ and $q^j_b$ the quantities available at the best ask and bid prices, respectively, for derivative $j$. Since we only consider trading at the best available prices, the position $x^j$ taken in derivative $j\in J$ is constrained to the interval $[-q^b_j, q^a_j]$. Equivalently, the chosen portfolio $x\in\reals^J$ has to belong to the Cartesian product
\begin{equation}\label{eq:D}
D:= \prod_{j \in J} [-q^j_b, q^j_a].
\end{equation}

Each quoted derivative has a finite number of payout dates. Options and futures have a single payout date while swap contracts pay periodically. Since the set $J$ of quoted derivatives is finite, the total number $I$ of payout dates is finite. We will denote the payout dates by $T_i$, $i=1,\ldots,I$. The payout at time $T_i$ from holding $x^j$ units of derivative $j$ will be denoted by $G^j_i(x^j)$. The functions $G^j_i$ are random. In case of a European option with maturity $T_k$, for example, we simply have
\[
G^j_i(x^j) = C^jx^j
\]
if $i=k$ and $G^j_i(x^j)=0$ otherwise. Here $C^j_i$ denotes the random payout of one unit of the option at maturity given by~\eqref{eq:callpo} or~\eqref{eq:putpo}, depending on whether the option is call or a put. If $j\in J$ is a futures contract, we have $s^j_a = s^j_b = 0$ but
$G^j_i(x^j)$ is a random nonlinear function of $x^j$ given by~\eqref{eq:3mfutpo}. The total payout of a portfolio $x\in\reals^J$ at time $T_i$ is given by 
\[
G_i (x) = \sum_{j\in J}G^j_i(x^j).
\]
Importantly, $G_i$ are {\em concave} functions of the portfolio $x$. In the computations below, we will only consider the best bid and ask prices, but the concavity would hold even if further price levels in the limit order book were taken into account; see e.g.~\cite{pen11}.

Consider an agent whose current financial position amounts to an obligation to pay $c_i$ units of cash at time $T_i$ for $i=0,\ldots,I$. We allow $c_i$ to be random but only to depend on the evolution of the underlying risk factors up to time $T_i$; see Section~\ref{sec:rf}. In other words, the sequence $c=(c_i)_{i=0}^I$ is adapted to the filtration $(\F_i)_{i=0}^I$ generated by the underlying stochastic process driving the risk factors. Here and in what follows, $\F_i$ denotes the $\sigma$-algebra generated by the risk factors observed by time $T_i$. We allow $c_i$ to take positive as well as negative values so they can describe both income and liability payments. In particular, $-c_0\in\reals$ can be interpreted as the initial cash endowment. 

We assume that all payments are rolled over to the terminal date $T_I$ using the overnight repo market. We assume, for simplicity, that the overnight market is perfectly liquid, so that the amount $z_i$ of cash invested at each time $T_i$ evolves according to the equation
\begin{equation}\label{eq:CashProcess90}
z_i = (1+R_i) z_{i-1} + G_i(x) - c_i \quad i = 1,\ldots, I.
\end{equation} 
Here, $x\in\reals^J$ is the portfolio of quoted derivative assets chosen at time $T_0$ and $R_i$ is the cumulative interest from the overnight market obtained over period $(T_{i-1},T_i]$. While investing in the overnight repo market yields the daily SOFR $r_t$, the money market return over the holding period $(T_{i-1},T_i]$ is given by
\[
R_i := R(T_{i-1}, T_{i}) \delta (T_i - T_{i-1})
\]
where $R(T_{i-1}, T_{i})$ is the geometric SOFR average in~\eqref{eq:SOFR_Average}. Since the future values of SOFR overnight rate $r_t$ are random, both $R(T_i, T_{i-1})$ and $R_i$ are random. We will assume that, $R_i\ge -1$ almost surely, i.e.\ that the average SOFR rate is always greater than $-100\%$. This implies that we can treat money as a ``freely disposable" asset. This property will be used in the next section to establish some crucial convexity properties of the trading problem. 

Equation~\eqref{eq:CashProcess90} implies that the terminal wealth is given by
\begin{equation}\label{eq:mmprocessroll}
z_I = z_0 \beta_0 + \sum_{i=1}^I[G_i(x) - c_i]\beta_i,
\end{equation}
where 
\begin{equation}\label{eq:rolling}
\beta_i := \prod_{j=i+1}^I (1+R_j)
\end{equation}
is the return accrued from rolling over cash in the money market from time $T_i$ to the terminal date $T_I$. The product over an empty set is defined as 1, so $\beta_I=1$. 

Since $G_i$, $R_i$ and $c_i$ are random, the terminal wealth $z_I$ is also random. We will denote the linear space of real-valued random variables by $L^0:=L^0(\Omega,\F,P)$. We will study the problem of optimizing the derivative portfolio $x\in\reals^J$ so that the random terminal net wealth distribution is as nice as possible as measured by an extended real-valued function $\V:L^0\to\ereals:=\reals\cup\{+\infty,-\infty\}$. The optimization problem can then be written as
\begin{equation}\tag{$OP$}\label{OP}
	\begin{aligned}
		&\minimize & &\V\left(z_0\beta_0 + \sum_{i=1}^I[G_i(x) - c_i] \beta_i\right)\quad\ovr\ x\in D,\ z_0\in\reals\\
		&\st & & z_0 + S_0(x) + c_0 \le 0 \\
	\end{aligned}
\end{equation}
where $D$ is the set given by~\eqref{eq:D} and $z_0$ denotes the amount of money market investments at the initial time $T_0$.

\begin{remark}\label{rem:margins}
CME SOFR futures and options are subject to margin requirements. One could incorporate margin requirements as additional constraints in problem~\eqref{OP}. A particularly interesting option would be to use the CME {\em portfolio margining service}. According to~\cite{CMEmargin}, the portfolio margining has shown capital efficiencies of up to 100\% for certain portfolios. For simplicity, and for the lack of a detailed description of the portfolio margining, we leave the study of the margin requirements for future work. For appropriate choices of the risk measure $\V$, however, one could expect the optimal solutions of~\eqref{OP} to be close to market-neutral and thus, to result in low capital requirements according to the portfolio margining.
\end{remark}

We will assume throughout that the function $\V:L^0\to\ereals$ is {\em convex} and {\em nonincreasing} in the partial order of $L^0$. The monotonicity simply means that the agent always prefers more money to less. The convexity corresponds to risk-aversion. The convexity and monotonicity imply that~\eqref{OP} is a {\em convex} optimization problem. Indeed, as already pointed out, the payout functions $G_i$ are concave in the portfolio~$x$. Being a composition of a concave $L^0$-valued function and a nonincreasing convex function $\V$ on $L^0$, the objective is convex; see e.g.~\cite{pen99}. The set $D$ is the Cartesian product of closed intervals and thus, convex. The cost function $S_0(x)$ in the budget constraint is convex; see~\cite{pen11}. It follows that the budget constraint specifies a convex region in the space of portfolios.

The budget constraint in~\eqref{OP} allows for throwing away cash. If we were to write it as an equality, the resulting problem would not be convex. Any optimal solution will, however, satisfy the constraint as an equality provided the function $\V$ is strictly decreasing, and interest rates stay strictly above $-100\%$. Indeed, the objective is then strictly decreasing in $z_0$. In the numerical experiments below, we will use the {\em entropic} risk measure $\V:L^0\to\ereals$ given by
\begin{equation}\label{eq:RiskMeasure}
	\V(u) \coloneqq \frac{1}{\rho} \ln E  \exp \left( -\frac{\rho u}{B_I} \right),
\end{equation}
where $\rho>0$ is a risk aversion parameter and $B_I$ a user-specified numeraire; see e.g.~\cite[Chapter~4]{fs16}. The entropic risk measure is convex and strictly decreasing on $L^0$.

The convexity of problem~\eqref{OP} greatly facilitates its numerical solution. Indeed, nonconvex optimization problems are, in general, NP-hard, while convex ones are of polynomial complexity; see e.g.~\cite{nes18}. We will solve problem~\eqref{OP} numerically by first approximating the expectation in~\eqref{eq:RiskMeasure} by a multivariate integration quadrature and then solving the resulting optimization problem by a convex optimization solver; see Section~\ref{sec:ipsd} below. The convexity is important also in derivative pricing in incomplete markets. In particular, it allows relating indifference prices to risk-neutral valuation and replication arguments; see Proposition~\ref{prop:ip} below. 

Problem~\eqref{OP} is inherently subjective. The user's risk preferences are captured by the risk measure $\V$, their financial position is described by the sequence $c$ of liability payments and their views are represented by the probability measure $P$ governing the future development of the underlying risk factors. Market conditions, in turn, are captured by the available quotes which are encoded in the functions $S_0$ and $G_i$. The subjective factors will also influence the indifference prices that build on problem~\eqref{OP}; see Section~\ref{sec:iph} below. This is quite natural since, in incomplete markets, traders are always exposed to some amount of risk when trading nonreplicable derivatives. Subjectivity is the driving force behind all trading in practice.

\section{The Optimum Value Function}\label{sec:optimum}

When the agent is involved in a trade, their financial position described by the sequence $c$ of cashflows in problem~\eqref{OP} is modified. We will denote the space of adapted sequences of cashflows by
\begin{equation*}
\M := \{(c_i)_{i=0}^I \midd c_i\in L^0(\Omega, \F_i, P)\ i=0,\ldots,I\}.
\end{equation*}
Here $L^0(\Omega,\F_i,P)$ denotes the linear space of equivalence classes of real-valued $\F_i$-measurable random variables. Given $c\in\M$, we will denote the optimum value of problem~\eqref{OP} by
\begin{equation}\label{eq:optimumterminal}
\varphi(c) := \inf_{x\in D, z_0\in\reals}\Biggl\{\V\left(z_0\beta_0 + \sum_{i=1}^I[G_i(x) -c_i] \beta_i\right)\ \bigg| \ z_0 + S_0(x) + c_0\le 0 \Biggr\}.
\end{equation}
The optimum value measures the quality of the agent's financial position as described by $c\in\M$. While the function $\V:L^0\to\ereals$ measures the ``risk" involved in delivering a random payment at time $T_I$, the function $\varphi:\M\to\ereals$ can be thought of as a market-consistent risk measure on sequences of payments. It takes into account the agent's views and risk preferences as well as the fact that, for given $c\in\M$, the agent can optimize his position in the derivatives market. The {\em optimum value function} $\varphi:\M\to\ereals$ will be instrumental in the pricing applications below.

To analyze the properties of $\varphi$, we start by writing problem~\eqref{OP} in terms of the {\em cash process} $z\in\M$ as 
\begin{equation}\label{OP2}
\begin{aligned}
&\minimize & &\V(z_I)\quad\ovr\ x\in D,\ z\in\M\\
&\st & & z_0 + S_0(x) + c_0 \le 0, \\
& & & z_i \le (1+R_i) z_{i-1} + G_i(x) - c_i\quad i=1,\ldots,I\quad P\text{-}a.s.
\end{aligned}
\end{equation}
In this problem, we have relaxed the equality~\eqref{eq:CashProcess90} to an inequality. Since the function $\V$ is decreasing with respect to the pointwise order of $L^0$, it is clear that this does not affect the optimal portfolio $x\in\reals^J$ nor the optimum value. Indeed, at optimal solutions, the inequality constraints will hold as equalities since we have assumed that interest rates remain strictly above $-100\%$; a rational agent never throws away money. The optimum value $\varphi(c)$ of~\eqref{OP} thus equals the optimum value of~\eqref{OP2}.

Introducing a dummy variable $d\in\M$, we can express the optimum value function in the form
\begin{align*}
\varphi(c) &= \inf_{x\in D, z\in\M}\{\V(z_I)\mid z_0 + S_0(x) + c_0 \le 0, \\
&\qquad\qquad\qquad\qquad  z_i \le (1+R_i) z_{i-1} + G_i(x) - c_i\ i=1,\ldots,I\ P\text{-}a.s.\}\\
&= \inf_{x\in D, z,d\in\M}\{\V(d_I-c_I)\mid z_I=0,\ c_i\le d_i,\ z_0 + S_0(x) + d_0 \le 0, \\
&\qquad\qquad\qquad\qquad\qquad z_i \le (1+R_i) z_{i-1} + G_i(x) - d_i\ i=1,\ldots,I\ P\text{-}a.s.\}
\end{align*}
or simply as 
\begin{equation}\label{varphi}
\varphi(c)=\inf_{d\in\C}\bar\V(d-c)
\end{equation}
where the function $\bar\V:\M\to\ereals$ is given by
\[
\bar\V(c):=\begin{cases}
\V(c_I) & \text{if $c_i\ge 0$ for $i=1,\ldots,I-1$ $P$-a.s.},\\
+\infty & \text{otherwise},
\end{cases}
\]
and
\begin{multline}\label{eq:sh}
\C:=\{c\in\M\mid \exists (x,z)\in D\times\M:\ z_0 + S_0(x) + c_0 \le 0,\\ 
z_i \le (1+R_i) z_{i-1} + G_i(x) - c_i\ i=1,\ldots,I,\ z_I=0\quad P\text{-}a.s.\}.
\end{multline}
The set $\C$ consists of claims $c\in\M$ that can be {\em superhedged} without a cost by trading the quoted hedging instruments and the overnight money market.

The {\em recession cone}
\begin{equation}\label{eq:rec}
\C^\infty := \{c\in\M\,|\,\bar c+\alpha c\in\C\quad \forall \bar c\in\C,\ \forall\alpha>0\}
\end{equation}
of $\C$ consists of the claims that can be superhedged without cost in unlimited amounts. The notion of recession cone (aka ``horizon cone" and ``asymptotic cone") is classic in convex analysis and optimization; see e.g.~\cite{roc70a,rw98}. Its importance in the study of nonlinear market models was pointed out in~\cite{pen14}; see also~\cite{pen11}. In traditional linear models or in models with proportional transaction costs without constraints, the set $\C$ is a convex cone so that $\C^\infty=\C$.

The following result shows how the convexity and monotonicity properties of the function $\V$ translate to convexity and monotonicity properties of the function $\varphi$. These properties will allow us, in particular, to relate indifference prices to risk-neutral valuations based on replication arguments; see Proposition~\ref{prop:ip} below.

\begin{proposition}\label{lem:varphi}
The set $\C\subset\M$ is convex and the function $\varphi:\M\to\ereals$ is convex with
\[
\varphi(\bar c+c)\le\varphi(\bar c)\quad\forall \bar c\in\M,\ c\in\C^\infty.
\]
\end{proposition}

\begin{proof}
To prove that $\C$ is convex, let $c^1,c^2\in\C$ and $\alpha^1,\alpha^2>0$ with $\alpha^1+\alpha^2=1$. By the definition of $\C$, there exists $(x^1,z^1)\in D\times\M$ that superhedges $c^1$ in the sense that
\begin{gather*}
z^1_0 + S_0(x^1) + c^1_0 \le 0,\\ 
z^1_i \le (1+R_i) z^1_{i-1} + G_i(x^1) - c^1_i\ i=1,\ldots,I,\ z^1_I=0\quad P\text{-}a.s.
\end{gather*}
Similarly, there exists $(x^2,z^2)\in D\times\M$ that superhedges $c^2$. The convexity of $S_0$ and the scenariowise concavity of $p_i$ now imply that the strategy $(\bar x,\bar z):=(\alpha^1x^1+\alpha^2x^2,\alpha^1z^1+\alpha^2z^2)$ superhedges $\bar c:=\alpha^1c^1+\alpha^2c^2$. Thus, $\bar c\in\C$, so $\C$ is convex.

Since the function $\V$ is convex on $L^0$, the function $\bar\V$ is convex on $\M$. This and the convexity of the set $\C$ imply that the function
\[
f(c,d):=\begin{cases}
\bar\V(d-c) & \text{if $d\in\C$},\\
+\infty & \text{otherwise}
\end{cases}
\]
is convex. Since $\varphi(c)=\inf_{d\in\M}f(c,d)$, by~\eqref{varphi}, the convexity of $\varphi$ thus follows from the general fact that an infimal projection of a convex function is convex; see e.g.~\cite[Theorem~1]{roc74}.

As to the last claim, let
$d\in\C$ and $c\in\C^\infty$. By the definition of the recession cone, we have $d+c\in\C$, and thus,
\[
\inf_{d\in\C}\bar\V(d-\bar c) \ge \inf_{d+c\in\C}\bar\V(d-\bar c) = \inf_{d'\in\C}\bar\V(d'-\bar c-c)
\]
which, by~\eqref{varphi},  means that $\varphi(\bar c+c)\le\varphi(\bar c)$.
\end{proof}

\section{Indifference Pricing}\label{sec:iph}

In an incomplete market, prices at which agents are willing to trade are best described by the classical indifference pricing principle; see e.g.~\cite{hodges1989optimal,car9,pen14} and the references there. Indifference pricing assumes that the agent optimizes their investment strategy before and after they buy or sell a financial product. The prices are defined as break-even prices at which the agent is indifferent about entering the trade. The resulting prices are naturally calibrated to the price quotes of the hedging instruments as well as the agent’s risk preferences, views and financial position; see~\cite{pen14} for further discussion and references. 


We will study indifference pricing of general swap contracts where two agents exchange a sequence $c\in\M$ (protection leg) for a constant multiple $\alpha\in\reals$ of another sequence $p\in\M$ (premium leg). The constant $\alpha$ will be called the {\em swap rate}. Traditional options where the premium is paid upfront, correspond to $p=(1,0,\ldots,0)$. In many contracts, the premium leg is a constant sequence, but in others, it is random. Examples of the latter include {\em credit default swaps}, where the premiums are paid only until the random default time. The premium sequence is random also in many kinds of {\em contracts for difference}.

The lowest swap rate at which a rational agent with an existing financial position $\bar c\in\M$ would be willing to take a short position in the swap is given by
\begin{equation}\label{eq:isp}
\pi_s(\bar c,p;c) := \inf\{\alpha\in \reals \midd \varphi(\bar c+c-\alpha p) \le \varphi(\bar c)\},
\end{equation}
where $\varphi:\M\to\ereals$ is the optimum value function of~\eqref{OP} introduced in Section~\ref{sec:optimum}. If the agent enters the short position with swap rate $\alpha$, their financial position would become $\bar c+c-\alpha p$. The inequality in~\eqref{eq:isp} means that, after reoptimizing their portfolio in the quoted derivatives, the quality of their net position, as measured by the optimum value of~\eqref{OP}, would be at least as good as it was before the trade. The value $\pi_s(\bar c,p;c)$ is called the {\em indifference swap rate} for taking a short position in the swap; see~\cite{pen14}. Similarly, the {\em indifference swap rate} for taking a long position in the swap is given by
\begin{equation}\label{eq:ibp}
\pi_b(\bar c,p;c) := \sup\{ \alpha \in \reals \midd \varphi(\bar c-c+\alpha p) \le \varphi(\bar c)\}.
\end{equation}

The indifference swap rates can be bounded by the {\em super-} and {\em subhedging swap rates} defined by
\begin{align*}
\pi_{\sup}(p;c) &= \inf\{\alpha\in\reals\,|\,c-\alpha p\in\C^\infty\}\\
\intertext{and}
\pi_{\inf}(p;c) &= \sup\{\alpha\in\reals\,|\,\alpha p-c\in\C^\infty\},
\end{align*}
where $\C^\infty$ is the {\em recession cone} of the set $\C$ of claims that can be superhedged without a cost; see~\eqref{eq:sh} and~\eqref{eq:rec}. Again, if $\C$ is a cone, as happens in traditional perfectly liquid market models and models with proportional transaction costs, the set $\C$ is a convex cone and $\C^\infty=\C$.

The following is essentially from~\cite[Section~4]{pen14}, but we repeat the simple argument for the convenience of the reader.

\begin{proposition}\label{prop:ip}
For any $\bar c,p\in\M$, the function $\pi_s(\bar  c,p; c)$ is convex as a function of $c\in\M$,
\[
\pi_b(\bar c, p; c) = -\pi_s(\bar  c,p; -c)\quad\forall c\in\M
\]
and, as soon as $\pi_s(\bar c,p;0)=0$,
\[
\pi_{\inf}(p;c)\le\pi_b(\bar c, p; c)\le \pi_s(\bar c, p; c)\le\pi_{\sup}(p;c)\quad\forall c\in\M.
\]
The above hold as equalities if, in addition, there exists an $\bar\alpha\in\reals$ such that both $c-\bar\alpha p$ and $\bar\alpha p-c$ belong to the set $\C^\infty$. In this case, the common value equals~$\bar\alpha$.
\end{proposition}

\begin{proof}
Let $\A(\bar c):=\{c\in\M\mid \varphi(\bar c+c) \le \varphi(\bar c) \}$.
We have
\[
\pi_s(\bar c,p;c) = \inf_{\alpha\in\reals} f(\alpha,c),
\]
where $f:\reals\times\M\to\ereals$ is given by
\[
f(\alpha,c) = \begin{cases}
\alpha & \text{if $c-\alpha p\in\A(\bar c)$},\\
+\infty & \text{otherwise}.
\end{cases}
\]
By Proposition~\ref{lem:varphi}, the function $\varphi$ is convex, so $\A(\bar c)$ is a convex set. This implies the convexity of $f$, so $\pi_s(\bar c,p;\cdot)$ is convex; see e.g.~\cite[Theorem~1]{roc74}. The equality in the statement follows directly from the definitions of the indifference swap rates. The convexity of $\pi_s(\bar c,p;\cdot)$ gives
\[
\pi_s(\bar c,p;0) = \pi_s(\bar c,p;\frac{1}{2}c+\frac{1}{2}(-c)) \le \frac{1}{2}\pi_s(\bar c,p;c) + \frac{1}{2}\pi_s(\bar c,p;-c)
\]
so if $\pi_s(\bar c,p;0)=0$, we have
\[
\pi_b(\bar c,p;c)=-\pi_s(\bar c,p;-c) \le \pi_s(\bar c,p;c).
\]
If $c-\alpha p\in\C^\infty$, Proposition~\ref{lem:varphi} gives $\varphi(\bar c+c-\alpha p)\le\varphi(\bar c)$ and thus 
\[
\pi_s(\bar c,p;c) = \inf\{\alpha\,|\,\varphi(\bar c+c-\alpha p)\le\varphi(\bar c)\} \le \inf\{\alpha\,|\, c-\alpha p\in\C^\infty\} = \pi_{\sup}(p;c).
\]
Since $\pi_{\inf}(p;c)=-\pi_{\sup}(p;-c)$ and $\pi_l(\bar c,p;c) = -\pi_s(\bar c,p;-c)$ we also get
\[
\pi_{\inf}(p;c)\le\pi_l(\bar c,p;c),
\]
which completes the proof of the inequalities. The last claim now follows from the fact that, if $c-\bar\alpha p\in\C^\infty$ and $\bar\alpha p-c\in\C^\infty$, then $\pi_{\sup}(\bar c,p;c)\le\bar\alpha$ and $\pi_{\inf}(p;c)\ge\bar\alpha$.
\end{proof}

The last part of Proposition~\ref{prop:ip} implies that the indifference rates are independent of the user's risk preferences and the financial position $\bar c$ provided the swap contract is such that $c-\bar\alpha p\in\C^\infty(-\C^\infty)$ for some $\bar\alpha\in\reals$. Indeed, the set $\C^\infty$ does not depend on such subjective factors. When $\C$ is a cone, we have $\C^\infty=\C$ so the above condition means that the sequences $c-\bar\alpha p$ and $\bar\alpha p-c$ can both be superhedged without a cost. In traditional linear models, this holds, in particular, when the sequences are replicable. The last part of Proposition~\ref{prop:ip} thus implies that, for replicable claims, indifference prices coincide with traditional risk-neutral prices.

Computing the indifference prices requires a one-dimensional search over $\alpha$ in~\eqref{eq:isp} and~\eqref{eq:ibp}. This can be done numerically provided that the optimum value $\varphi(c)$ of problem~\eqref{OP} can be evaluated for a given sequence of cash flows $c\in\M$. Since an exact evaluation of $\varphi$ is rarely possible, it is natural to redefine $\varphi$ as the lowest attainable value. The resulting indifference prices then depend not only on an agent's views, risk preferences and the current financial position, but also on their expertise in optimising their trading strategies. 

The next result says that, when $\V$ is the entropic risk measure given by~\eqref{eq:RiskMeasure}, the computation of the two indifference prices reduces to the solution of just three instances of problem~\eqref{OP}, provided the random variable
\begin{equation}\label{eq:B}
Z:=\sum_{i=0}^Ip_i\beta_i
\end{equation}
is almost surely strictly positive and one chooses $B_I=Z$ in~\eqref{eq:RiskMeasure}. Here $\beta_i$ is the cumulative interest earned in the repo market over $[T_i,T_I]$; see~\eqref{eq:rolling}.

\begin{proposition}\label{prop:swapopt}
If the risk measure $\V$ is given by~\eqref{eq:RiskMeasure} where $B_I=Z$ is strictly positive, then the indifference swap rates are given by
\begin{align*}
\pi_s(\bar c , p; c)  = \varphi(\bar c + c) - \varphi(\bar c),\\
    \pi_b(\bar c , p; c)  = \varphi(\bar c) - \varphi(\bar c-c).
\end{align*}
\end{proposition}

\begin{proof}
By the definition of the optimum value function $\varphi$,
\begin{multline*}
    \varphi(\bar c + c - \alpha p) 
    = \inf_{x\in D, z_0\in \reals}\Biggl\{ \V\left(z_0\beta_0 + \sum_{i=1}^IG_i(x) \beta_i - \sum_{i=1}^I(\bar c_i + c_i - \alpha p_i)\beta_i\right) \\
\bigg|\ z_0 + S_0(x) + \bar c_0 + c_0 - \alpha p_0\le 0\  P\text{-}a.s. \Biggr\}.
\end{multline*}
The change of variables, $\tilde z_0:=z_0-\alpha p_0$, gives
\begin{multline}\label{eq:Z}
    \varphi(\bar c+c-\alpha p) 
    = \inf_{x\in D, \tilde z_0\in \reals}\Biggl\{\V\left((\tilde z_0+\alpha p_0)\beta_0 + \sum_{i=1}^IG_i(x) \beta_i - \sum_{i=1}^I(\bar c_i + c_i - \alpha p_i)\beta_i\right) \\
\bigg|\ \tilde z_0 + S_0(x) + \bar c_0 + c_0 \le 0\  P\text{-}a.s. \Biggr\}.
\end{multline}
The definition of $\V$ in~\eqref{eq:RiskMeasure} gives
\begin{align*}
\V\Bigg((\tilde z_0+\alpha p_0)\beta_0 &+ \sum_{i=1}^IG_i(x) \beta_i - \sum_{i=1}^I(\bar c_i + c_i - \alpha p_i)\beta_i\Bigg)\\
&=\V\left(\tilde z_0\beta_0 + \sum_{i=1}^IG_i(x) \beta_i - \sum_{i=1}^I(\bar c_i + c_i)\beta_i + \alpha\sum_{i=0}^Ip_i\beta_i\right)\\
&=\frac{1}{\rho}\ln E\exp\left[-\frac{\rho}{B_I}\left(\tilde z_0\beta_0 + \sum_{i=1}^IG_i(x) \beta_i - \sum_{i=1}^I(\bar c_i + c_i)\beta_i\right) - \rho\alpha\right]\\
&=\V\left(\tilde z_0\beta_0 + \sum_{i=1}^IG_i(x) \beta_i - \sum_{i=1}^I( \bar c_i+c_i)\beta_i\right) - \alpha.
\end{align*}
Substituting this into~\eqref{eq:Z} gives
\[
\varphi(\bar c + c - \alpha p) = \varphi(\bar c+c) - \alpha.
\]
Using this in the definition of the indifference selling price in~\eqref{eq:isp} gives the expression for $\pi_s(\bar c,p;c)$ in the statement. The expression for $\pi_b(\bar c,p;c)$ follows by a similar argument.
\end{proof}
 
Under the assumptions of  Proposition~\ref{prop:swapopt}, we have, in particular, $\pi_s(\bar c,p;0) = 0$, which by Proposition~\ref{prop:ip}, implies $\pi_b(\bar c,p;c)\le\pi_s(\bar c,p;c)$ for all $c\in\M$.

\section{Pricing and Hedging SOFR Derivatives}\label{sec:ipsd}

This section computes the indifference swap rates and the corresponding hedging strategies for SOFR swaps, swaptions and caplets using the pricing formulas in Proposition~\ref{prop:swapopt}. Such derivatives cannot be replicated, so traditional risk-neutral pricing would be difficult to justify for trading purposes.

To evaluate the optimum value function and to find the optimal hedging portfolios, problem~\eqref{OP} is solved by first approximating the expectation by an integration quadrature and then solving the resulting problem using the interior point solver of MOSEK Optimizer~\cite{aps2019mosek}. All computations are performed using an Intel$^{\text{\sffamily\textregistered}}$ Core$^{\text{\sffamily{\tiny TM}}}$ i7 at 3.00 GHz~x~8 Dell Latitude laptop running Linux with 16 GB of RAM and 128 GB of virtual memory. The problem is implemented in Python~\cite{van1995python} using NumPy~\cite{harris2020array}, Pandas~\cite{mckinney2015pandas} and Numba~\cite{lam2015numba} to construct the integration quadratures and the payout matrices. The problem is then communicated to the interior point solver using MOSEK Optimizer API~\cite{aps2019mosek}. In all the examples below, the indifference prices are computed in less than a minute.

\subsection{Overnight Index Swaps}\label{sec:ois}

As mentioned at the beginning of Section~\ref{sec:otc}, payments of a SOFR swap can be replicated by taking appropriate positions in zero-coupon bonds and rolling the payments in the overnight market. This becomes impractical when ZCBs are illiquid. Bickersteth et al.~\cite{BDR2025} construct dynamic replicating strategies using futures contracts. However, their market model is complete, and it does not allow for jumps in the overnight rate. We will compute indifference prices and optimal hedges for overnight index swaps using the model from the previous sections.

We start by studying two-month SOFR Overnight Index Swaps (OIS) having a notional of \$500,000 on the 28 August 2024. Since the maturity is less than a year, all cashflows are paid at maturity so that the premium process $p$ is given by $p_i = 0$ for $i < I$ and $p_I = 1$, whereas the sequence $c$ is given by $c_i = 0$ for $i < I$ and $c_I$ is swap floating leg; see Section~\ref{sec:exchange}. We choose $B_I=Z=p_I\beta_I=1$ in \eqref{eq:B}, so the indifference swap rates are given by the formulas in Proposition~\ref{prop:swapopt}. The number of the available hedging instruments for the two-month investment horizon is 47. They consist of a June three-month SOFR futures contract and call and put options on the three-month September futures. While the June futures expires on the third Wednesday of September 2024, the options expire on the preceding Friday. We assume that the user's initial financial position consists of \$1 million in cash. This means that the process $\bar c$ in~\eqref{eq:isp} and~\eqref{eq:ibp} is given by $\bar c_0 = -10^6$ and $\bar c_i = 0$ for $i>0$. We start by choosing the risk aversion parameter $\rho=100$. 

The sell rate is 5.23544\% while the buy rate is 5.23177\%. Figure~\ref{fig:swaphs} illustrates the hedging portfolios for a short (left-hand plot) and long (right-hand plot) positions in the two-month SOFR OIS. The hedging portfolios are defined as the difference between the optimal portfolios before and after entering the swap position. The strategies are sparse, consisting of only about 7 out of the 47 available derivatives. 

Figure~\ref{fig:swaphspo} compares the payouts of the hedging strategies with the payouts of the OIS being priced. The blue dots give the payouts of the hedging portfolio against the value of the floating leg in different scenarios. The orange curves give the payouts of the swaps. Exact replication of the OIS is not possible with the considered 47 hedging instruments, but the payout of the constructed hedging strategy approximately tracks the OIS payout. The hedging strategy is optimized to the given quotes as well as to the agent's views and risk preferences.

\begin{figure}[!ht] 
	\centering
	{   
		\includegraphics[trim = 0mm 0mm 0mm 0mm, clip, width=0.45\textwidth] 
        {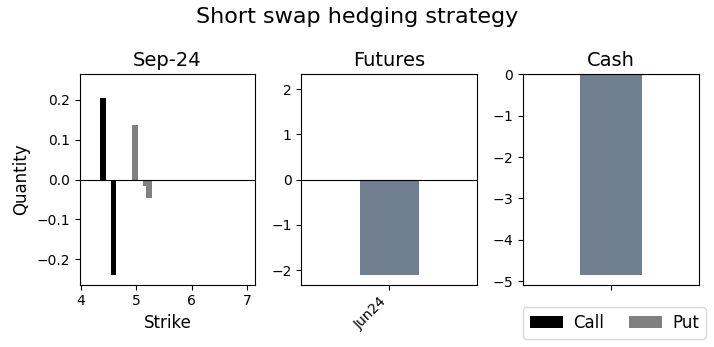}
	}
    {   
		\includegraphics[trim = 0mm 0mm 0mm 0mm, clip, width=0.45\textwidth] 
        {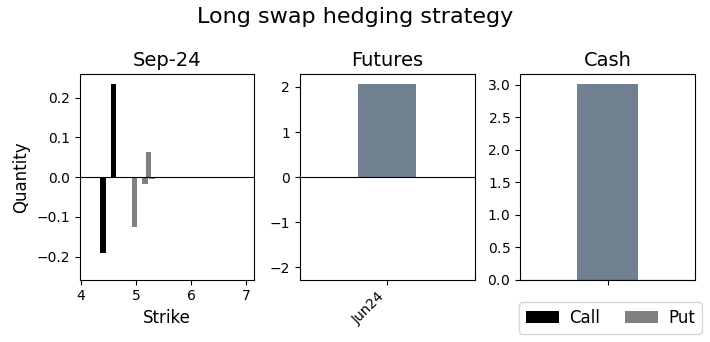}
	}
	\centering
	\caption
	{Hedging strategies for a two-month SOFR OIS for a short position (left) and a long position (right) obtained with a risk aversion parameter $\rho=100$ and a notional of \$500,000 on 28 August 2024. The OIS indifference sell rate is 5.23544\% while the buy rate is 5.23177\%. Cash is expressed in USD. 
	}
	\label{fig:swaphs}
\end{figure}

\begin{figure}[!ht] 
	\centering
	{  
		\includegraphics[trim = 0mm 0mm 0mm 0mm, clip, width=0.45\textwidth] 
        {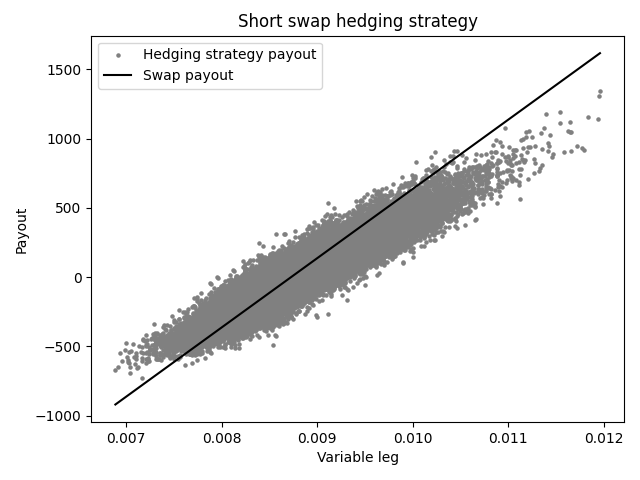}
	}
    {  
		\includegraphics[trim = 0mm 0mm 0mm 0mm, clip, width=0.45\textwidth] 
        {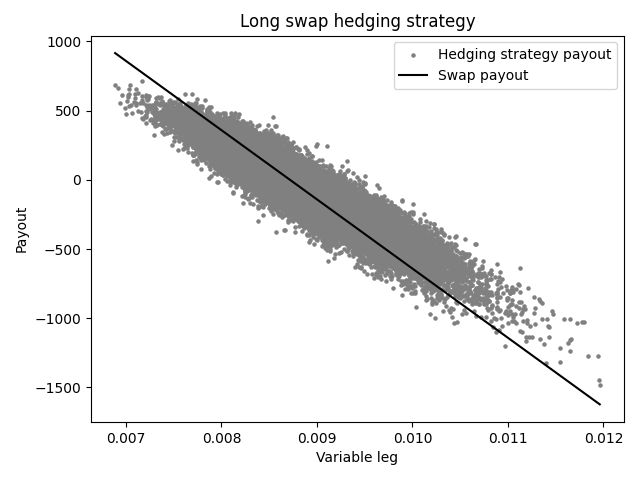}
	}
	\centering
	\caption 
	{The payouts of the hedging strategies (in Figure~\ref{fig:swaphs}) compared with the payout of a two-month SOFR OIS as a function of the OIS variable leg $R(T_0,T)\delta(T-T_0)$ for a short position (left) and a long position (right) on 28 August 2024. The OIS notional is \$500,000 while the risk aversion parameter is $\rho=100$. The OIS indifference sell rate is 5.23544\% while the buy rate is 5.23177\%. The payout is expressed in USD, while the variable leg is expressed in decimals.
	}
	\label{fig:swaphspo}
\end{figure}  

We will next study the sensitivity of the indifference swap rates with respect to the risk aversion parameter and transaction costs. We compute the indifference rates for varying values of the risk aversion parameter $\rho$ while keeping all other problem parameters and specifications unchanged. The left-hand plot in Figure~\ref{fig:rhosensswap} shows the indifference buy and sell rates for the two-month OIS as function of the risk aversion parameter. As $\rho$ increases, the rates decrease and the spread between the buy and sell rates first widens before starting to narrow at approximately $4.83$. The observed swap sell and buy rates on a Bloomberg terminal on 28 August 2024 were 5.14789\% and 5.14217\%, respectively. These quotes correspond to a risk aversion of roughly 3000.

\begin{figure}[!ht] 
	\centering
    \subfloat
	{   
		\includegraphics[trim = 0mm 0mm 0mm 0mm, clip, width=0.49\textwidth] 
        {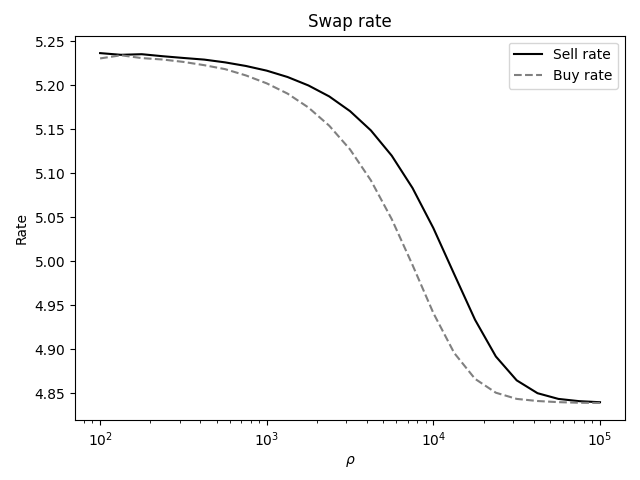}
	}
    \subfloat
	{   
		\includegraphics[trim = 0mm 0mm 0mm 0mm, clip, width=0.49\textwidth] 
    {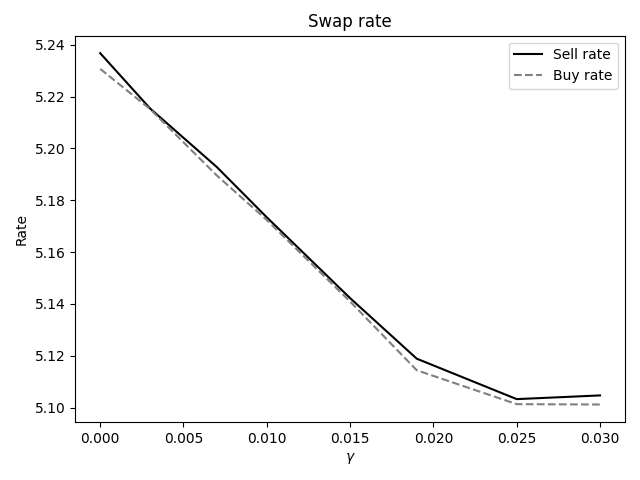}
	}
	\centering
	\caption
	{Two-month indifference swap rates as a function of the risk aversion parameter $\rho$ (left-hand plot) and transaction costs $\gamma$ (right-hand plot) on 28 August 2024. The risk aversion level is $\rho=100$ for the right-hand plot. Rates and transaction costs are expressed as percentages.
	}
	\label{fig:rhosensswap}
\end{figure}

We next study the effect of market liquidity on the indifference prices. More precisely, we compute the indifference rates after adding proportional transaction costs with an identical rate $\gamma$ across all derivatives quotes. The dependence of the indifference rates on $\gamma$ is illustrated in the right-hand plot of Figure~\ref{fig:rhosensswap}. Both rates seem to decrease linearly up to around $\gamma=0.025\%$ after which the prices remain nearly constant. For transaction costs above $\gamma=0.025\%$, the hedging portfolio consists solely of the money market account.

\subsection{Swaptions}\label{sec:swaption}

Yueh and Wu~\cite{yueh24} discuss the challenges of pricing and hedging SOFR swaptions. Section~4.5 of~\cite{BDR2025} gives a replication strategy using SOFR futures, but again, their model is a complete market model that doesn't allow for jumps in the rates.

We will study the indifference prices on 28 August 2024 of a call swaption having a strike of 3\%, a notional of \$500,000, and expiring in four months. The underlying OIS expires one year after the call maturity. Entering a SOFR swaption requires an upfront premium payment at the time $T_0$ of purchase. The premium process $p$ is given by $p_0 = 1$ and $p_i=0$ for $i > 0$, and the sequence of cash flows $c$ is given by $c_i = 0$ for $i < I$ while $c_I$ is the payout of the swaption given by~\eqref{eq:swaptionpo4} with $T_I$ equal to the swaption expiry date. We choose $B_I=Z= p_0\beta_0$ in~\eqref{eq:B}, so the swaption prices are given by Proposition~\ref{prop:swapopt}. For the four-month investment horizon, 116 hedging instruments are available. They consist of three-month futures and options, expiring in December and September 2024; see Section~\ref{sec:exchange}. 
As in Section~\ref{sec:ois}, we assume that the user's initial financial position consists of \$1 million in cash, i.e.~the process $\bar c$ in~\eqref{eq:isp} and~\eqref{eq:ibp} is given by $\bar c_0 = -10^6$ and $\bar c_i = 0$ for $i>0$. We start by choosing the risk aversion parameter $\rho=100$. 

The indifference selling price is \$243.11 while the buying price is \$232.62. The corresponding hedging strategies are illustrated in Figure~\ref{fig:swaptionhs}. Again, the hedging portfolios are sparse, consisting of only about 10 instruments out of the 116 available ones. Figure~\ref{fig:swaptionhspo} compares the payouts of the hedging strategies (blue dots) against the swaption payouts (orange curves). Exact replication of the swaption payout is not possible with the considered 116 hedging instruments. The constructed hedges are optimal given the quotes of the hedging instruments and the specified views and risk preferences.

\begin{figure}[!ht] 
	\centering
    \subfloat
	{   
		\includegraphics[trim = 0mm 0mm 0mm 0mm, clip, width=0.49\textwidth] 
        {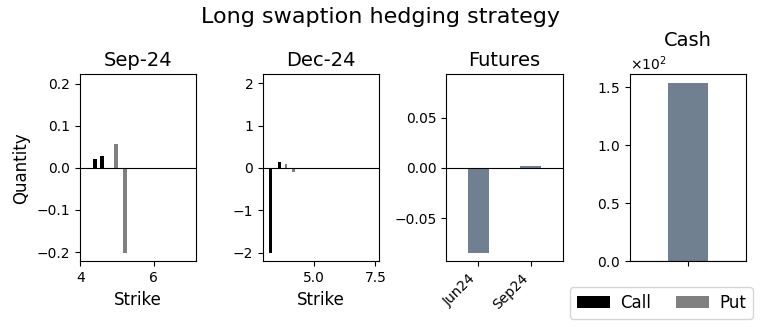}
	}
    \subfloat
    {  
		\includegraphics[trim = 0mm 0mm 0mm 0mm, clip, width=0.49\textwidth] 
        {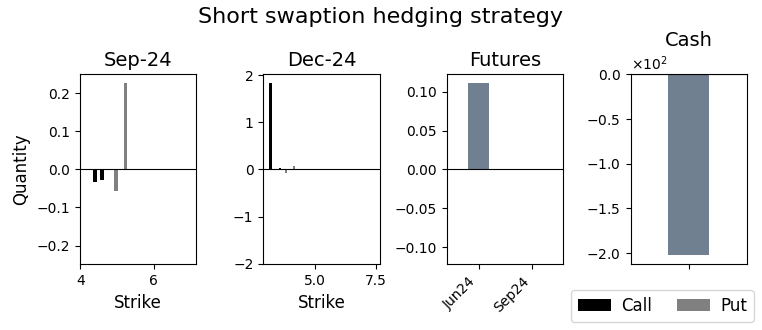}
	}
	\centering
	\caption
	{  Hedging strategies on 28 August 2024 of long (left) and short (right) positions in a call swaption with strike 3\% expiring in four months and an underlying OIS expiring in one year. The notional is \$500,000 while the risk aversion parameter is $\rho=100$. The indifference selling price is \$243.11 while the buying price is \$232.62. Cash is expressed in USD.
	}
	\label{fig:swaptionhs}
\end{figure}

\begin{figure}[!ht] 
	\centering
    \subfloat
	{   
		\includegraphics[trim = 0mm 0mm 0mm 0mm, clip, width=0.49\textwidth] 
        {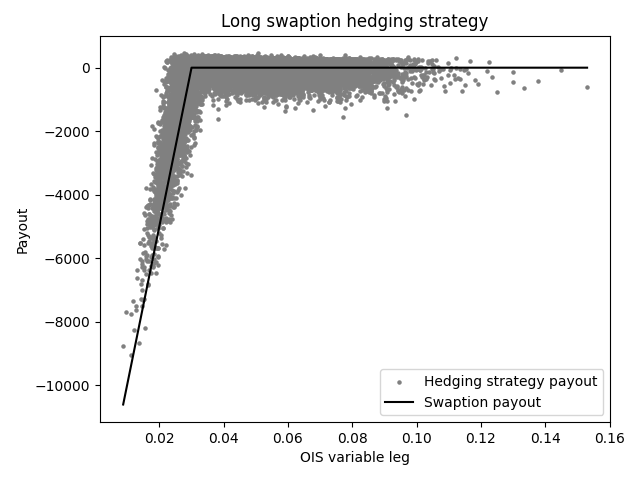}
	}
    \subfloat
    {   
		\includegraphics[trim = 0mm 0mm 0mm 0mm, clip, width=0.49\textwidth] 
        {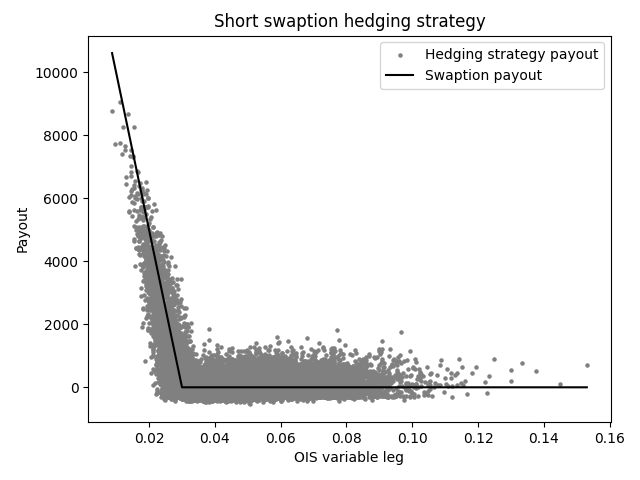}
	}
	\centering
	\caption
	{Payouts of the hedging strategies (in Figure~\ref{fig:swaptionhs}) compared with the call swaption payout as a function of the underlying OIS variable leg for a long position (left) and short position (right). The notional is \$500,000 while the risk aversion parameter is $\rho=100$. The indifference selling price is \$243.11 while the buying price is \$232.62. The payout is expressed in USD, whereas the underlying OIS variable leg is expressed in decimals.
	}
	\label{fig:swaptionhspo}
\end{figure}

We will next study the dependence of indifference prices on the risk-aversion parameter by computing the prices for varying levels of risk aversion parameter $\rho$. The resulting prices are illustrated in the left-hand plot of Figure~\ref{fig:rhosensswaption}. Unlike in the case of the OIS in Section~\ref{sec:ois}, the spread between the indifference buying and selling prices increases monotonically with $\rho$.

\begin{figure}[!ht]
	\centering
    \subfloat
	{   
		\includegraphics[trim = 0mm 0mm 0mm 0mm, clip, width=0.49\textwidth] 
        {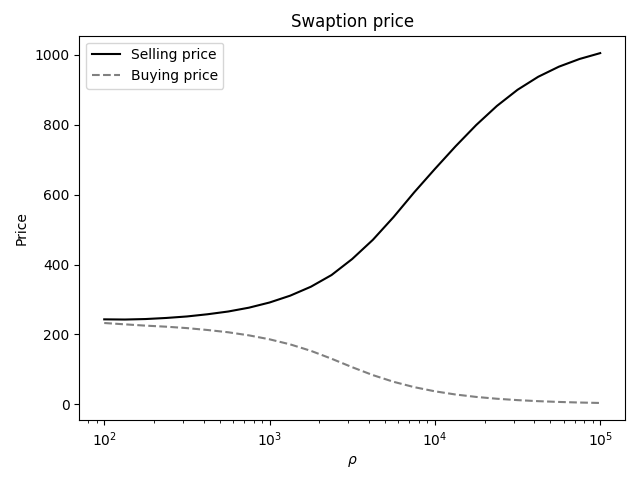}
	}
    \subfloat
	{ 
		\includegraphics[trim = 0mm 0mm 0mm 0mm, clip, width=0.49\textwidth] 
        {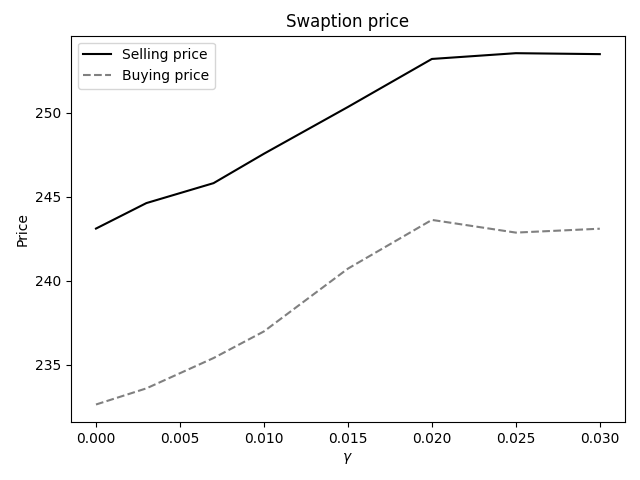}
	}
	\centering
	\caption
	{Price sensitivity with respect to the risk aversion parameter $\rho$ (left-hand plot) and transaction costs $\gamma$ (right-hand plot) on 28 August 2024 of a four-month swaption with strike 3\% and an underlying OIS expiring in one year. The risk aversion level is $\rho=100$ for the right-hand plot. Price is expressed in USD, while transaction costs $\gamma$ are expressed as percentages.
	}
	\label{fig:rhosensswaption}
\end{figure}

Right-hand plot in Figure~\ref{fig:rhosensswaption} illustrates the dependence of the indifference prices on the liquidity of the hedging instruments. As in Section~\ref{sec:ois}, we increase the bid-ask spreads of the hedging instruments by adding proportional transaction costs $\gamma$ to each instrument. 
The right-hand plot in Figure~\ref{fig:rhosensswaption} plots the indifference prices as a function of $\gamma$. Both prices appear to increase up to approximately $\gamma=0.02\%$, after which the prices remain nearly constant. For transaction costs above $\gamma=0.02\%$, the hedging portfolio consists solely of the money market account.

\subsection{Caps and Floors}\label{sec:caps}

Risk-neutral pricing together with hedging of caps and floors have been studied e.g.\ in~\cite{BDR2025,fontana2024term}. While the former assumed a complete market model and provided replication strategies using futures contracts, the latter employs a term structure model with stochastic discontinuities. Such a model gives rise to incompleteness, which makes replication impossible, in general. To address this,~\cite{fontana2024term} constructs approximate hedging strategies using local risk-minimization aka.\ quadratic hedging; see~\cite{MS2001}.

We study the indifference prices on 28 August 2024 of a six-month caplet whose reference period starts in one month, has a strike of 3\% and a notional of \$500,000. The caplet premium is paid upfront. The premium process $p$ is given by $p_0 = 1$ and $p_i=0$ for $i>0$, and the sequence of cashflows $c$ is given by $c_i=0$ for $i<I$ while $c_I$ is the caplet payout given by~\eqref{eq:capletpo} with $T_I$ equal to the caplet expiry date. We choose $B_I=Z=p_0\beta_0$ in~\eqref{eq:B}, so the caplet prices are given by Proposition~\ref{prop:swapopt}. For the seven-month investment horizon, 207 hedging instruments are available. They consist of three-month futures and options, expiring in March 2025 and December and September 2024; see Section~\ref{sec:exchange}. 
As in Section~\ref{sec:ois}, we assume that the user's initial financial position consists of \$1 million in cash, i.e.~the process $\bar c$ in~\eqref{eq:isp} and~\eqref{eq:ibp} is given by $\bar c_0 = -10^6$ and $\bar c_i = 0$ for $i>0$. We start by choosing the risk aversion parameter $\rho=100$. 

The indifference selling price is \$96.35 while the buying price is \$94.78. The corresponding hedging strategies are illustrated in Figure~\ref{fig:capleths}. Again, the hedging portfolios are sparse, consisting of about 20 instruments out of the 207 available ones. Figure~\ref{fig:caplethspo} compares the payouts of the hedging strategies (blue dots) against the caplet payouts (orange curves). Exact replication of the caplet payout is not possible with the considered 207 hedging instruments. The constructed hedges are optimal given the quotes of the hedging instruments and the specified views and risk preferences.

\begin{figure}[!ht] 
	\centering
    \subfloat
	{   
		\includegraphics[trim = 0mm 0mm 0mm 0mm, clip, width=0.49\textwidth] 
        {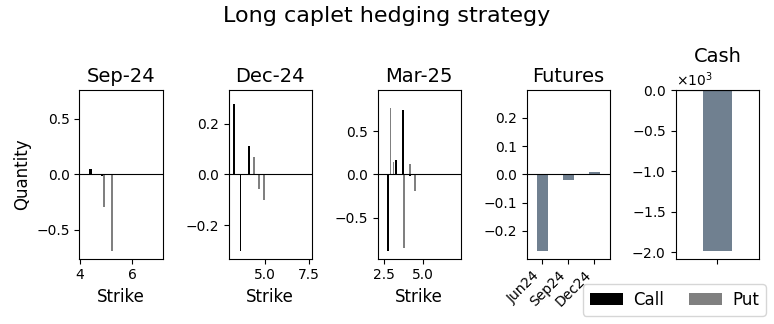}
	}
    \subfloat
    {  
		\includegraphics[trim = 0mm 0mm 0mm 0mm, clip, width=0.49\textwidth] 
        {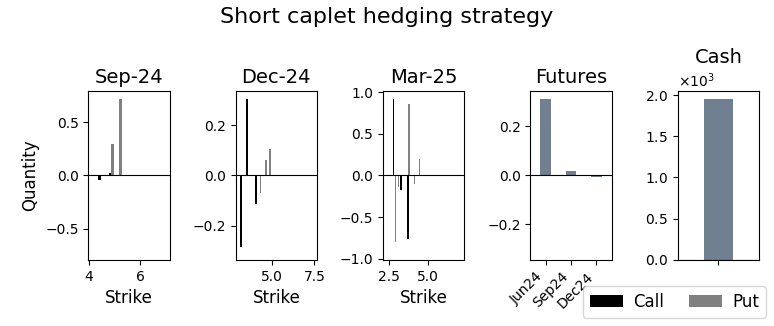}
	}
	\centering
	\caption
	{Hedging strategies on 28 August 2024 of long (left) and short (right) positions in a six-month SOFR caplet starting in one month with strike 3\% expiring in six months. The notional is \$500,000 while the risk aversion parameter is $\rho=100$. The indifference selling price is \$96.35 while the buying price is \$94.78. Cash is expressed in USD.
	}
	\label{fig:capleths}
\end{figure}

\begin{figure}[!ht] 
	\centering
    \subfloat
	{   
		\includegraphics[trim = 0mm 0mm 0mm 0mm, clip, width=0.49\textwidth] 
        {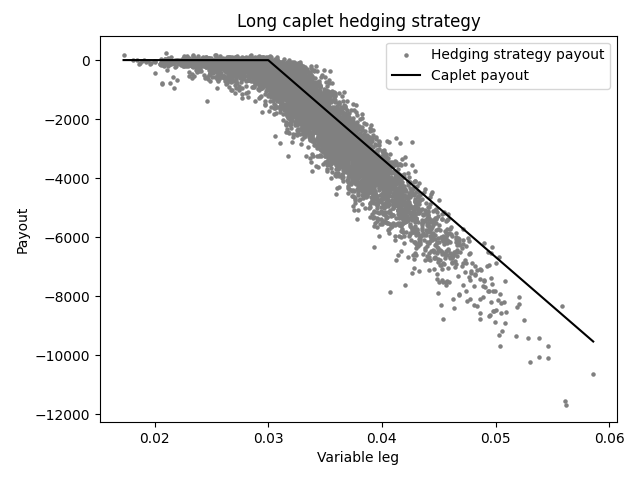}
	}
    \subfloat
    {   
		\includegraphics[trim = 0mm 0mm 0mm 0mm, clip, width=0.49\textwidth] 
        {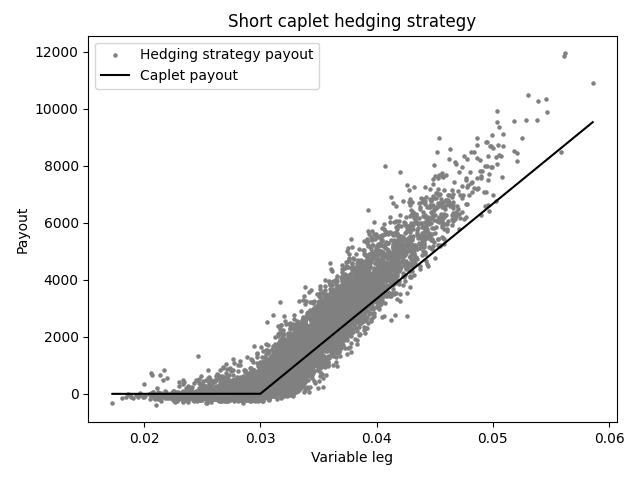}
	}
	\centering
	\caption
	{Payouts of the hedging strategies (in Figure~\ref{fig:capleths}) compared with the caplet payout as a function of the caplet variable leg for a long position (left) and short position (right). The notional is \$500,000 while the risk aversion parameter is $\rho=100$. The indifference selling price is \$96.35 while the buying price is \$94.78. The payout is expressed in USD, whereas the caplet variable leg is expressed in decimals.
	}
	\label{fig:caplethspo}
\end{figure}

We will next study the dependence of indifference prices on the risk-aversion parameter by computing the prices for varying levels of risk aversion parameter~$\rho$. The resulting prices are illustrated in the left-hand plot of Figure~\ref{fig:senscaplet}. Like in the case of the swaption in Section~\ref{sec:swaption}, the spread between the indifference buying and selling prices increases monotonically with $\rho$.

Right-hand plot in Figure~\ref{fig:senscaplet} illustrates the dependence of the indifference prices on the liquidity of the hedging instruments. As in the previous sections, we increase the bid-ask spreads of the hedging instruments by adding proportional transaction costs $\gamma$ to each instrument. 
The right-hand plot in Figure~\ref{fig:senscaplet} plots the indifference prices as a function of $\gamma$. Both prices appear to increase up to approximately $\gamma=0.35\%$, after which the prices remain nearly constant. For transaction costs above $\gamma=0.35\%$, the hedging portfolio consists solely of the money market account.

\begin{figure}[!ht]
	\centering
    \subfloat
	{   
		\includegraphics[trim = 0mm 0mm 0mm 0mm, clip, width=0.49\textwidth] 
        {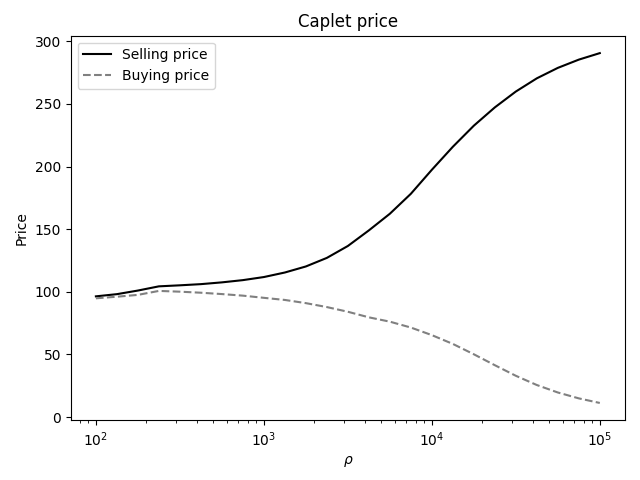}
	}
    \subfloat
	{ 
		\includegraphics[trim = 0mm 0mm 0mm 0mm, clip, width=0.49\textwidth] 
        {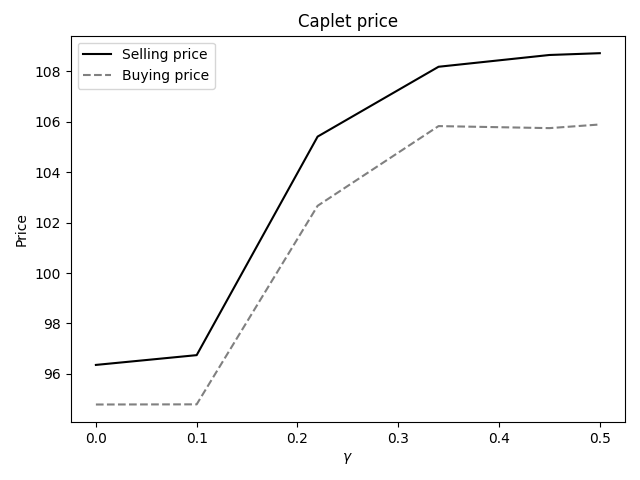}
	}
	\centering
	\caption
	{Price sensitivity with respect to the risk aversion parameter $\rho$ (left-hand plot) and transaction costs $\gamma$ (right-hand plot) on 28 August 2024 of a six-month SOFR caplet starting in one month with strike 3\%. The risk aversion level is $\rho=100$ for the right-hand plot. Price is expressed in USD, while transaction costs $\gamma$ are expressed as percentages.
	}
	\label{fig:senscaplet}
\end{figure}

\bibliographystyle{plain}
\bibliography{sp}

\end{document}